\documentclass[12pt,reqno,twoside]{article}
\usepackage{amsfonts, amsmath, amssymb, amscd, amsthm}

\usepackage{multicol}
\usepackage{comment}
\usepackage{color} 

\newtheorem{lemma}{Lemma}
\newtheorem{thm}{Theorem}[section]
\newtheorem{proposition}{Proposition}
\newtheorem{rmk}{Remark}[section]
\newtheorem{exm}{Example}[section]
\newcommand{\res}{\mathrm{res}}
\newcommand{\e}{\mathrm{e}}
 \numberwithin{equation}{section}

\numberwithin{equation}{section}

\newcommand{\pa}{\partial}
\newcommand{\ld}{\lambda}

\title{Bilinear identities for an extended B-type Kadomtsev-Petviashvili hierarchy}

\author{Runliang Lin $^{a}$ \footnote{corresponding author, e-mail: rlin@math.tsinghua.edu.cn},
Tiancheng Cao $^{a}$ \footnote{e-mail: ctc12@mails.tsinghua.edu.cn},
Xiaojun Liu $^{b}$ \footnote{e-mail: xjliu@cau.edu.cn}
and Yunbo Zeng $^{a}$ \footnote{e-mail: yzeng@math.tsinghua.edu.cn}}

\date{}

\begin{document}

\maketitle

$^{a}${\small\textit{Department of Mathematical Sciences, Tsinghua University, Beijing 100084, P.R. China}}

$^{b}${\small\textit{Department of
  Applied Mathematics, China Agricultural University, Beijing 100083, P.R. China}}

\begin{abstract}
In this paper, we construct the bilinear identities for the wave functions of an extended B-type Kadomtsev-Petviashvili (BKP) hierarchy, which contains two types of (2+1)-dimensional Sawada-Kotera equation with a self-consistent source (2d-SKwS-I and 2d-SKwS-II). By introducing an auxiliary variable corresponding to the extended flow for the BKP hierarchy, we find the tau-function and the bilinear identities for this extended BKP hierarchy. The bilinear identities can generate all the Hirota's bilinear equations for the zero-curvature forms of this extended BKP hierarchy.
As examples, the Hirota's bilinear equations for the two types of 2d-SKwS (both 2d-SKwS-I and  2d-SKwS-II) will be given explicitly.
\end{abstract}


\medskip\noindent
{\it 2010 Mathematics Subject Classification: 35Q51, 37K10}

\maketitle

\

\centerline{(Published in: ``Theoretical and Mathematical Physics, 186(3) (2016) 307--319'')}

\section{Introduction}
The well-known Kadomtsev-Petviashvili hierarchy \cite{1983-Date-Jimbo-etal-BKP,Dickey} is an infinite-dimensional system of nonlinear partial differential equations, which contains various types of soliton equations. The B-type Kadomtsev-Petviashvili hierarchy \cite{DateRes,1983-Date-Jimbo-etal-BKP} (or BKP hierarchy for short) is a sub-hierarchy for the KP system. The BKP hierarchy possesses many integrable structures (see, e.g., \cite{ChengHeHu2010,1983-Date-Jimbo-etal-BKP,Shen} and the references therein), such as Lax formulation, $\tau$-function, Hirota's bilinear equations, etc.

Sato's theory has fundamental importance in the study of integrable systems.
It reveals the infinite dimensional Grassmannian structure of the space of $\tau$-functions, where the $\tau$-functions are solutions to the Hirota's bilinear form of KP hierarchy. Bilinear identity, which is a bilinear residue identity for wave and adjoint wave functions, plays an important role in the proof of existence for $\tau$-function, and it is also the generating function of the Hirota's bilinear equations for the KP hierarchy \cite{ZhangYJ,1997-Loris-Willox,Shen}.

The reductions and generalizations are important topics in the study of integrable systems. Several approaches for the generalizations have been developed, e.g., constructing new
flows to \emph{extend} original systems. There are several ways to introduce new flows to make a new
compatible integrable system. In \cite{Xiong}, the KP hierarchy is extended by properly
combining additional flows. In \cite{Kamata2002}, extension of KP hierarchy is formulated by
introducing fractional order pseudo-differential operators. In \cite{Dimakis2004a,Dimakis2004b},
Dimakis and M\"uller-Hoissen extended the Moyal-deformed hierarchies by including additional evolution
equations with respect to the deformation parameters. In \cite{Carlet2003}, Carlet, Dubrovin and
Zhang defined a logarithm of the difference Lax operator and got the extended (2+1)D Toda lattice
hierarchy. Later, the Hirota's bilinear formalism and the relations of extended (2+1)D Toda lattice
hierarchy and extended 1D Toda lattice hierarchy have been studied
\cite{Milanov2007,Takasaki2010}. In \cite{Li2010}, Li, {\it et al}, studied the $\tau$-functions and
bilinear identities for the extended bi-graded Toda lattice hierarchy.
In \cite{Willox2014},
a discrete analogue of the symmetry constrained KP hierarchy is presented.

The soliton equations with self-consistent sources have many physical applications (see \cite{2015-Chvartatskyi-Sydorenko,Grinevich-Taimanov-2008,Lin2001,Lin2006,Lin2010,Melnikov1983,Melnikov1989} and the references therein). For example,
 the speed of the solitons can be changed by the sources (see, e.g., \cite{Lin2001,Melnikov1989}). Recently, it is shown that the construction of the soliton equation with self-consistent sources is related to the ``\emph{squared eigenfunction symmetry}''
and the ``\emph{binary Darboux transformation}'' of the original soliton hierarchy \cite{2014-Doliwa-Lin,Lin2013}.
In \cite{Wu}, an extended BKP hierarchy is constructed following the idea in \cite{Liu} by using the symmetries generating functions (or the squared eigenfunction symmetries), where two kinds of new multi-component BKP hierarchy are constructed,  their $n$-reduction and $k$-constraint are discussed. This kind of extended BKP hierarchy can be thought of as a generalization of the BKP hierarchy by the squared eigenfunction symmetries. The extended BKP hierarchy obtained in \cite{Wu} includes two types of (2+1)-dimensional Sawada-Kotera equation with a self-consistent source (2d-SKwS-I and 2d-SKwS-II), where the 2d-SKwS-I coincides with that obtained in \cite{Hu} by a source generating method.
In \cite{Liu}, we introduced a method to construct an extended KP hierarchy,
then the bilinear identities for this extended KP hierarchy is constructed in \cite{Lin2013}. Hirota's bilinear equations are derived for all the zero-curvature forms in the extended KP hierarchy by introducing an auxiliary flow. It seems that the Hirota's bilinear equations given in \cite{Lin2013} are in a simpler form by comparing with the previous results in \cite{HuKP}.

In this paper, we will construct the bilinear identities for the extended BKP hierarchy, which is introduced in \cite{Wu}. The bilinear identities will generate all the Hirota's bilinear forms for the zero-curvature equations of this extended BKP hierarchy. As examples, we will derive the Hirota bilinear forms for the two types of (2+1)-dimensional Sawada-Kotera equation with a self-consistent source (2d-SKwS-I and 2d-SKwS-II).
To the best of our knowledge,
the Hirota's bilinear equations for the 2d-SKwS-II did not appear in the literatures before.

This paper is organized as follows. In Section \ref{Bilinear identities for BKP hierarchy}, the bilinear identities for the BKP hierarchy with a squared eigenfunction symmetry (or a ``ghost symmetry'') are constructed. In Section \ref{Bilinear identities for extended BKP hierarchy}, by considering the squared eigenfunction symmetry as an auxiliary flow, we  construct the bilinear identities for the extended BKP hierarchy, we also  prove that these bilinear identities fully characterize the extended BKP hierarchy. In Section \ref{tau functions and Hirota bilinear equations for extended BKP hierarchy}, we introduce the $\tau$-function for the extended BKP hierarchy and we find the generating functions for Hirota's bilinear form for the extended BKP hierarchy. In Section \ref{Transformation from Hirota's bilinear to nonlinear form}, we  transform the Hirota's bilinear form back to the nonlinear partial differential equations for the 2d-SKwS-I and 2d-SKwS-II cases, which verifies the correctness of our construction. In the last section, conclusion and remarks are given.

\section{Bilinear identities for the BKP hierarchy with its ``ghost symmetry''}
\setcounter{equation}{0}
\label{Bilinear identities for BKP hierarchy}

We recall some notation on a pseudo-differential operator.
For a pseudo-differential operator
 $P=\sum\limits_{i=-\infty}^n a_i\partial^i$ with $\partial:=\partial_x$,
we denote its non-negative part, negative part and adjoint operator by $P_+=\sum\limits_{i=0}^n a_i\partial^i$,
$P_-=\sum\limits_{i<0}^n a_i\partial^i$,
$P^*=\sum\limits_{i=-\infty}^n (-\partial)^i a_i$, respectively,
and $\res_\partial (P)=a_{-1}$.
For a Laurent series $f(\lambda)=\sum_i c_i\lambda^i$,
we denote $\res_\lambda f(\lambda)=c_{-1}$.

Consider a pseudo-differential operator
$L:=\partial+\sum\limits_{i=1}^{\infty}u_i \partial^{-i}$,
and define $B_n:=(L^n)_+$,
then the BKP hierarchy is defined by \cite{1983-Date-Jimbo-etal-BKP}
\begin{subequations}
\label{BKP}
\begin{align}
& L_{t_n}=[B_n ,L ], \qquad n=1,3,5,...  \label{BKP:a}
\end{align}
with the constraint
\begin{align}
& L^*=-\partial L \partial^{-1}. \label{BKP:b}
\end{align}
\end{subequations}

The constrained condition (\ref{BKP:b}) is equivalent to $B_n\cdot 1=0,\ \ n\in\mathbb{N}_{odd}$ (see \cite{DateRes,1983-Date-Jimbo-etal-BKP}).
Following the idea in \cite{Oevel} and \cite{Lin2013},
we can introduce a new $\partial_z$-flow as
\begin{subequations}
\label{extBKP}
\begin{align}
&\partial_z L=[ r \partial^{-1} q_x -q \partial^{-1} r_x ,L ], \label{extBKP:a}
\end{align}
where $q$ and $r$ satisfy
\begin{align}
&q_{t_n}=B_n(q), \qquad n=1,3,5,..., \label{extBKP:b}\\
&r_{t_n}=B_n(r). \label{extBKP:c}
\end{align}
\end{subequations}

Here $q$ and $r$ satisfying (\ref{extBKP:b}) and (\ref{extBKP:c}) are called eigenfunctions \cite{1997-Loris-Willox,1999-Loris-Willox},
while adjoint eigenfunctions $q^*$ and $r^*$ satisfy $q^*_{t_n}=-B^*_n(q^*),\ r^*_{t_n}=-B^*_n(r^*)$.
It is easy to see that $q_x$ and $r_x$ can be adjoint eigenfunctions, then the compatibility of the $\pa_z$-flow (\ref{extBKP:a})
and the $\pa_{t_n}$-flows (\ref{BKP:a}) are ensured.
The $\partial_z$-flow describes a symmetry for BKP hierarchy, which is called the
{\it squared eigenfunction symmetry} \cite{Oevel} or a {\it ghost flow} \cite{ghost,ghost2}.

We can introduce dressing operator $W=1+\sum\limits_{i=1}^{\infty}w_i(z,t)\partial^{-i}$
where $t=(t_1\equiv x,t_3,t_5,\cdots)$
as $L=W\partial W^{-1}$, where $W$ satisfies $t_n-$evolution equations as
$$\partial_{t_n}W=-(W\partial^nW^{-1})_{-}W=-L^n_-W,~~n=1,3,5,\cdots.$$
Then by constraint condition (\ref{BKP:b}), we have $W^*\partial W=\partial$.
The dressing operator $W$ also satisfies the following equation:
\begin{equation}
\label{Wz}
\partial_z W=( r \partial^{-1} q_x -q \partial^{-1} r_x )W.
\end{equation}

We define the wave and adjoint wave function as
\begin{subequations}
\label{wave}
\begin{align}
w(z,t,\lambda)&=W\e^{\xi(t,\lambda)} ,\\
w^*(z,t,\lambda)&=(W^*)^{-1}\e^{-\xi(t,\lambda)},
\end{align}
\end{subequations}
where $\xi(t,\lambda)=\sum\limits_{n=0}^\infty t_{2n+1}\lambda^{2n+1}$.
These wave functions satisfy
\begin{subequations}
\label{propertiesofw}
\begin{align}
&Lw(z,t,\lambda)=\lambda w(z,t,\lambda), &\partial_{t_n}w(z,t,\lambda)=B_n w(z,t,\lambda),\\
&L^{*}w^*(z,t,\lambda)=\lambda w^*(z,t,\lambda), &\partial_{t_n}w^*(z,t,\lambda)=-B_n^* w(z,t,\lambda).
\end{align}
\end{subequations}

We recall an important lemma here (see \cite{Dickey} for details and proof)
which is useful for proving the propositions in this paper.
\begin{lemma}
\label{lemma}
Let $P$ and $Q$ are pseudo-differential operators, then
\begin{equation}
\label{lem}
\res_\partial P\cdot Q^*=\res_\lambda P(\e^{\xi(t,\lambda)})\cdot Q(\e^{-\xi(t,\lambda)}).
\end{equation}
\end{lemma}

\begin{thm}
\label{res-pro}
The BKP hierarchy with a squared eigenfunction symmetry ((\ref{BKP}) and (\ref{extBKP})) is equivalent to the following residue identities
\begin{subequations}
\label{res}
\begin{align}
&\res_\lambda \lambda^{-1}w(z,t,\lambda)w(z,t^{'},-\lambda)=1 \label{res:a},\\
&\res_\lambda \lambda^{-1}w_z (z,t,\lambda)w(z,t^{'},-\lambda)=q(z,t)r(z,t^{'})-r(z,t)q(z,t^{'}) \label{res:b},\\
&q(z,t)=\res_\lambda \lambda^{-1}w(z,t,\lambda)\cdot\partial^{-1}_{x^{'}}( q(z,t^{'})w_{x^{'}}(z,t^{'},-\lambda)) \label{res:c},\\
&r(z,t)=\res_\lambda \lambda^{-1}w(z,t,\lambda)\cdot\partial^{-1}_{x^{'}}( r(z,t^{'})w_{x^{'}}(z,t^{'},-\lambda)) \label{res:d}.
\end{align}
\end{subequations}
where the inverse of $\pa$ is understood as pseudo-differential operator acting on an exponential
  function, e.g., $\pa^{-1}\left(r w_x\right)=(\pa^{-1}r \pa W) (e^{\xi(t,\ld)})$.
\end{thm}

\begin{proof}
The bilinear identity (\ref{res:a}) of BKP hierarchy has already been proved in \cite{DateRes,1983-Date-Jimbo-etal-BKP}
where $z$ can be regarded as a fixed parameter here.
By (\ref{extBKP:b}), (\ref{extBKP:c}), and (\ref{propertiesofw}), we know that
the mixed partial derivatives $\partial_{t_1}^{m_1}\partial_{t_3}^{m_3}\cdots\partial_{t_k}^{m_k}$ ($k\in\mathbb{N}_{odd}$) on $w(z,t,\lambda)$, $q(z,t)$ and $r(z,t)$
have the same expression in terms of a differential operator $P_{m_1\cdots m_k}$
\begin{align*}
&\partial_{t_1}^{m_1}\partial_{t_3}^{m_3}\cdots\partial_{t_k}^{m_k}w(z,t,\lambda)
=P_{m_1\cdots m_k}(w(z,t,\lambda)),\\
&\partial_{t_1}^{m_1}\partial_{t_3}^{m_3}\cdots\partial_{t_k}^{m_k}q(z,t)=P_{m_1\cdots m_k}(q(z,t)),\\
&\partial_{t_1}^{m_1}\partial_{t_3}^{m_3}\cdots\partial_{t_k}^{m_k}r(z,t)=P_{m_1\cdots m_k}(r(z,t)),
\end{align*}
Notice that $P_{m_1\cdots m_k}$ does not contain constant term.
Then to prove (\ref{res:b}), we only need to show
\begin{align*}
&\res_\lambda \lambda^{-1}w_z(z,t,\lambda)\partial_{t_1}^{m_1}\partial_{t_3}^{m_3}\cdots\partial_{t_k}^{m_k}w(z,t,-\lambda) \\
=&\res_\lambda (r\partial^{-1}q_x-q\partial^{-1}r_x)W\e^{\xi(t,\lambda)}P_{m_1\cdots m_k}W(-\partial^{-1})\e^{-\xi(t,\lambda)}\\
=&\res_\partial (r\partial^{-1}q_x-q\partial^{-1}r_x)W\partial^{-1}W^*P^*_{m_1\cdots m_k} \\
=&\res_\partial (r\partial^{-1}q_x-q\partial^{-1}r_x)\partial^{-1}P^*_{m_1\cdots m_k} \\
=&qP_{m_1\cdots m_k}\partial^{-1}(r_x)-rP_{m_1\cdots m_k}\partial^{-1}(q_x)
=qP_{m_1\cdots m_k}(r)-rP_{m_1\cdots m_k}(q).
\end{align*}
This proves (\ref{res:b}).

Notice that for any $P_{m_1\cdots m_k}$ introduced above,
we have
\begin{align*}
& \res_\lambda \partial_{t_1}^{m_1}\partial_{t_3}^{m_3}\cdots\partial_{t_k}^{m_k}
\lambda^{-1}w(z,t,\lambda)\cdot\partial^{-1}(q(z,t)w_{x}(z,t,-\lambda))\\
=&\res_\lambda P_{m_1\cdots m_k}(\lambda^{-1}w(z,t,\lambda))\cdot\partial^{-1}(q(z,t)w_{x}(z,t,-\lambda))\\
=&\res_\lambda P_{m_1\cdots m_k}W\partial^{-1}\e^{\xi(t,\lambda)}\cdot\partial^{-1}q(z,t)\partial W\e^{-\xi(t,\lambda)}\\
=&\res_\partial P_{m_1\cdots m_k}W\partial^{-1}\cdot(\partial^{-1}q(z,t)\partial W)^*\\
=&\res_\partial P_{m_1\cdots m_k}W\partial^{-1}W^*(-\partial)q(z,t)(-\partial)^{-1}\\
=&\res_\partial P_{m_1\cdots m_k}q(z,t)\partial^{-1}=P_{m_1\cdots m_k}(q(z,t))
=\partial_{t_1}^{m_1}\partial_{t_3}^{m_3}\cdots\partial_{t_k}^{m_k} q(z,t),
\end{align*}
Therefore, we have
$$q(z,t^{'})=\res_\lambda \lambda^{-1}w(z,t^{'},\lambda)\cdot\partial^{-1}(q(z,t)w_{x}(z,t,-\lambda)),$$
and similarly
$$r(z,t^{'})=\res_\lambda \lambda^{-1}w(z,t^{'},\lambda)\cdot\partial^{-1}(r(z,t)w_{x}(z,t,-\lambda)).$$
Then (\ref{res:c}) and (\ref{res:d}) are proved by interchanging $t$ and $t^{'}$.

  The proof for the inverse part of this theorem is written as the following proposition.
\end{proof}

\begin{proposition}
\label{zBKP}
If functions $q(z,t)$, $r(z,t)$ and wave function
$$w(z,t,\lambda)=W\e^{\xi(t,\lambda)},
\qquad W=(1+\sum\limits_{i\geq 1}w_i(z,t)\lambda^{-i}),$$
satisfy the residue identities (\ref{res}),
then the pseudo-differential operator
$L=W\partial W^{-1}$,
$q$ and $r$ are a solution of BKP hierarchy with squared eigenfunction symmetry ((\ref{BKP}) and (\ref{extBKP})).
\end{proposition}

\begin{proof}
The equation (\ref{BKP}) can be derived from (\ref{res:a}) (see \cite{DateRes,1983-Date-Jimbo-etal-BKP}).
The (\ref{extBKP:b}) and (\ref{extBKP:c}) can be proved by taking
$\partial_{t_n}$ on (\ref{res:c}) and (\ref{res:d}) respectively.
Equation (\ref{extBKP:a}) can be proved by (\ref{res:b}) as the following.


It is easy to show that $(W_z W^{-1})_+=0$.
Notice that the adjoint wave function $w^*(z,t,\lambda)$
can be written in the following way
    $$w^*(z,t,\lambda):=(W^*)^{-1}\e^{-\xi(t,\lambda)}=\partial W\partial^{-1}\e^{-\xi(t,\lambda)}=-\lambda^{-1}w_x(z,t,-\lambda),$$
then from (\ref{res:b}), we have
    $$\res_\lambda w_z(z,t,\lambda)w^*(z,t^{'},\lambda)
    =r(z,t)q_{x^{'}}(z,t^{'})-q(z,t)r_{x^{'}}r(z,t^{'}).$$
Furthermore, we have
\begin{align*}
    \res_\partial W_z W^{-1}\partial^m &=\res_\lambda W_z \e^{\xi(t,\lambda)}(-\partial)^m W^{*-1}\e^{-\xi(t,\lambda)}
    &=\res_\lambda w_z(z,t,\lambda)(-\partial)^m w^*(z,t,\lambda)\\
    &=r(-\partial)^m q_x-q(-\partial)^m r_x,
\end{align*}
which means
\begin{equation*} W_zW^{-1}=\sum_{m=0}^{\infty}(r(-\partial)^m(q_x)-q(-\partial)^m(r_x))\partial^{-m-1}
    =r\partial^{-1}q_x-q\partial^{-1}r_x.
\end{equation*}
Hence (\ref{Wz}) holds, which implies (\ref{extBKP:a}).

\end{proof}

\section{Bilinear identities for an extended BKP hierarchy}
\label{Bilinear identities for extended BKP hierarchy}

In \cite{Wu}, an extended BKP hierarchy is defined
by using the squared eigenfunction symmetry as (\ref{extBKP:a}).
Two types of BKP hierarchy with self-consistent sources are found,
their Lax representations are also given.

We recall the extended BKP hierarchy here (for a fixed odd $k$ and $k\neq 1$)
\begin{subequations}
\label{extended}
\begin{align}
&\partial_{\bar{t_k}} L=[ B_k+r \partial^{-1} q_x -q \partial^{-1} r_x ,L ] ,\label{extended:b}\\
&L_{t_n}=[B_n ,L ] ,\quad  L^{*}=-\partial L\partial^{-1}, \quad (n\neq k,n=1,3,5,...)\label{extended:a}\\
&q_{t_n}=B_n(q),\label{extended:c}\\
&r_{t_n}=B_n(r),\label{extended:d}
\end{align}
\end{subequations}
where $q$ and $r$ are eigenfunctions.
This hierarchy is constructed by replacing an arbitrary fixed $k$-th flow $\partial_{t_k}$ by $\partial_{\bar{t_k}}$ where the $\partial_{\bar{t_k}}$-flow is a linear combination of $\partial_{t_k}$- and $\partial_z$-flow.

\begin{rmk}
  For simplicity of the notions, we will still use the symbols $w(z,t,\ld)$, $w^*(z,t,\ld)$,
  $q(z,t)$, $r(z,t)$, $L$, $W$, etc., in this and the following sections, but they should be understood to
  the case of the extended BKP hierarchy (\ref{extended}).
  For example, from now on,
  $t= (t_1,t_3,\cdots,t_{k-2},\bar{t_k},t_{k+2},\cdots)$ ($k$ is odd).
\end{rmk}

In \cite{Shen,Shen2}, the Hirota's bilinear equations for constrained BKP hierarchy (\ref{extended}) are constructed.
A natural question is to find the bilinear identities
for the extended BKP hierarchy (\ref{extended}),
because the bilinear identities provide a systematic way to generate
all the Hirota bilinear equations
in the extended BKP hierarchy (\ref{extended}).
In this section, we give a detailed construction on how to derive the bilinear identities for the extended BKP hierarchy (\ref{extended}).

Dressing operator $W$ is given by $W=1+\sum\limits_{i=1}^{\infty}w_i(z,t)\partial^{-i}$,
where $t=(t_1\equiv x,t_3,\cdots,t_{k-2},\bar{t_k},t_{k+2},\cdots)$,
it satisfies
    $$\partial_{\bar{t_k}}W=-L_-^k W+(r \partial^{-1} q_x -q \partial^{-1} r_x)W.$$
The wave function and its adjoint are defined as (\ref{wave})
except for the change:
$\xi(t,\lambda)=\bar{t_k}\lambda^k+\sum\limits_{n\neq k}t_n\lambda^n,$
($k$ and $n$ are odd).
The wave and adjoint wave functions satisfy
\begin{subequations}
\label{propertiesofwtk}
\begin{align}
&Lw(z,t,\lambda)=\lambda w(z,t,\lambda),~~~ \partial_{t_n}w(z,t,\lambda)=B_n w(z,t,\lambda),~~(n\neq k)\\
&L^{*}w^*(z,t,\lambda)=\lambda w^*(z,t,\lambda),~ \partial_{t_n}w^*(z,t,\lambda)=-B_n^* w(z,t,\lambda).
\end{align}
\end{subequations}

Then we have the bilinear identities for the extended BKP hierarchy
as the following proposition.

\begin{proposition}
The bilinear identities for the extended BKP hierarchy (\ref{extended}) are given by the following sets of residue identities with an auxiliary variable $z$:
\begin{subequations}
\label{res-extented}
\begin{align}
&\res_\lambda \lambda^{-1}w(z-\bar{t_k},t,\lambda)w(z-\bar{t_k}^{'},t^{'},-\lambda)=1 ,\label{res-extented:a}\\
&\res_\lambda \lambda^{-1}w_z (z-\bar{t_k},t,\lambda)w(z-\bar{t_k}^{'},t^{'},-\lambda) \nonumber\\
&\qquad
=q(z-\bar{t_k},t)r(z-\bar{t_k}^{'},t^{'})-r(z-\bar{t_k},t)q(z-\bar{t_k}^{'},t^{'}) ,
\label{res-extented:b}\\
&q(z-\bar{t_k},t)=\res_\lambda \lambda^{-1}w(z-\bar{t_k},t,\lambda)\cdot\partial^{-1}_{x^{'}} (q(z-\bar{t_k}^{'},t^{'})w_{x^{'}}(z-\bar{t_k}^{'},t^{'},-\lambda)) ,\label{res-extented:c}\\
&r(z-\bar{t_k},t)=\res_\lambda \lambda^{-1}w(z-\bar{t_k},t,\lambda)\cdot\partial^{-1}_{x^{'}} (r(z-\bar{t_k}^{'},t^{'})w_{x^{'}}(z-\bar{t_k}^{'},t^{'},-\lambda)) .\label{res-extented:d}
\end{align}
\end{subequations}
\end{proposition}

\begin{proof}
Notice that
$\frac{d}{d\bar{t_k}}w(z-\bar{t_k},t,\lambda)
=(\partial_{\bar{t_k}}-\partial_z)
w(z,t,\lambda)|_{z=z-\bar{t_k}}
=B_k(w(z-\bar{t_k},t,\lambda)),$
so (\ref{res-extented:a}) can be proved as the original BKP case \cite{DateRes,1983-Date-Jimbo-etal-BKP}.


The proofs of (\ref{res-extented:b})-(\ref{res-extented:d}) are  similar to the proofs of (\ref{res:b})-(\ref{res:d}), just recalling that $\frac{d}{d\bar{t_k}}w(z-\bar{t_k},t,\lambda)=B_k(w(z-\bar{t_k},t,\lambda)).$
\end{proof}

\begin{proposition}
Suppose
    $$w(z,t,\lambda)=W\e^{\xi(t,\lambda)},
    \qquad W=(1+\sum\limits_{i\geq 1}w_i(z,t)\lambda^{-i}),$$
$q(z,t)$ and $r(z,t)$ satisfy the bilinear identities (\ref{res-extented}),
then the pseudo-differential operator $L=W\partial W^{-1}$, $q$ and $r$
satisfy the extended BKP hierarchy (\ref{extended}).
\end{proposition}

\begin{proof}
 It is already known \cite{DateRes} that
(\ref{res-extented:a}) implies the constraint  $W^*\partial W=\partial$
or equivalently $L^*=-\partial L\partial^{-1}$ (\ref{extended:a}).
We may define adjoint wave function of $w(z,t,\lambda)$ as $w^*(z,t,\lambda)=(W^{-1})^*\e^{-\xi(t,\lambda)}$,
then by $W^*\partial W=\partial$ we get $w^*(z,t,\lambda)=-\lambda^{-1}w_x(z,t,-\lambda)$.
Similar to the proof of Proposition \ref{zBKP}, we obtain
$W_z=(r\partial^{-1}q_x-q\partial^{-1}r_x)W$ from (\ref{res-extented:b}).

It can be derived that
$\frac{d}{d\bar{t_k}}w(z-\bar{t_k},t,\lambda)=\left( \frac{d}{d\bar{t_k}}W+L^kW \right)\e^{\xi(t,\lambda)}$
and
\begin{equation*}
\res_\lambda w(z-\bar{t_k},t,\lambda)\lambda^{-1}w_{x^{'}}(z-\bar{t_k}^{'},t^{'},-\lambda)
=-\res_\lambda w(z-\bar{t_k},t,\lambda)w^*(z-\bar{t_k}^{'},t^{'},\lambda)
=0.
\end{equation*}
Then from the coefficients of Taylor expansion of the above equation
we find that for any positive integer $m$
\begin{eqnarray*}
    0&=&\res_\lambda \frac{d}{d\bar{t_k}}w(z-\bar{t_k},t,\lambda)\cdot\partial^m w^*(z-\bar{t_k},t,\lambda)\\
    &=&\res_\lambda \left( \frac{d}{d\bar{t_k}}W+L^kW \right)\e^{\xi(t,\lambda)}\cdot\partial^m (W^{-1})^*\e^{-\xi(t,\lambda)}\\
    &=&\res_\partial \left( \frac{d}{d\bar{t_k}}W+L^kW \right)\cdot\left(\partial^m (W^{-1})^*\right)^*\\
    &=&\res_\partial \left( \frac{d}{d\bar{t_k}}W+L^k_-W \right)W^{-1}(-\partial)^m, \qquad (m>0)
\end{eqnarray*}
which means that $\frac{d}{d\bar{t_k}}W=-L^k_-W$.
Hence $\partial_{\bar{t_k}}W=-L_-^k W+(r \partial^{-1} q_x -q \partial^{-1} r_x)W$ holds, then (\ref{extended:b}) can be derived.

For (\ref{extended:c}) and (\ref{extended:d}), the proofs can be done by directly differentiating (\ref{res-extented:c}) and (\ref{res-extented:d}), since we have already known (\ref{extended:a}) and hence (\ref{propertiesofwtk}) follows.

\end{proof}

\section{$\tau$ function for the extended BKP hierarchy}
\label{tau functions and Hirota bilinear equations for extended BKP hierarchy}

The existence of $\tau$-function for original BKP hierarchy is proved in \cite{1983-Date-Jimbo-etal-BKP}. According to \cite{ZhangYJ} and  \cite{Lin2013},
 similar assumptions can be made
for the extended BKP hierarhcy (\ref{extended}), i.e.,
\begin{subequations}
\label{tau}
\begin{align}
&w(z-\bar{t_k},t,\lambda)=\frac{\tau(z-\bar{t_k}+\frac{2}{k\lambda^k},t-2[\lambda])}{\tau(z-\bar{t_k},t)}\cdot \e^{\xi(t,\lambda)},\label{tau-w}\\
&q(z,t)=\frac{\sigma(z,t)}{\tau(z,t)},\label{tau-q}\\
&r(z,t)=\frac{\rho(z,t)}{\tau(z,t)},\label{tau-r}
\end{align}
where $[\lambda]=(\frac{1}{\lambda},\frac{1}{3\lambda^3},\frac{1}{5\lambda^5},\cdots)$.
\end{subequations}

Similar to \cite{Shen}, notice that $w^*(z-\bar{t_k},t,\lambda)=-\lambda^{-1}w_x(z-\bar{t_k},t,-\lambda)$,
we have
\begin{subequations}
\label{Phi}
\begin{align}
&\partial^{-1}( r(z-\bar{t_k},t)w^*(z-\bar{t_k},t,\lambda) )\nonumber\\
=&-\frac{\rho(z-\bar{t_k},t)\tau(z-\bar{t_k}-\frac{2}{k\lambda^k},t+2[\lambda])
+\rho(z-\bar{t_k},t+2[\lambda])\tau(z-\bar{t_k},t)}{2\lambda\tau^2(z-\bar{t_k},t)}\e^{-\xi(t,\lambda)},\\
&\partial^{-1}( q(z-\bar{t_k},t)w^*(z-\bar{t_k},t,\lambda) )\nonumber\\
=&-\frac{\sigma(z-\bar{t_k},t)\tau(z-\bar{t_k}-\frac{2}{k\lambda^k},t+2[\lambda])
+\sigma(z-\bar{t_k},t+2[\lambda])\tau(z-\bar{t_k},t)}{2\lambda\tau^2(z-\bar{t_k},t)}\e^{-\xi(t,\lambda)}.
\end{align}
\end{subequations}

Substituting (\ref{tau}) and (\ref{Phi}) into (\ref{res-extented}), one can get
\begin{subequations}
\label{barhirota}
\begin{align}
&\res_\lambda \lambda^{-1}\bar{\tau}(t-y-2[\lambda])\bar{\tau}(t+y+2[\lambda])\e^{\xi(-2y,\lambda)}=\bar{\tau}(t-y)\bar{\tau}(t+y) ,\\
&\bar{\sigma}(t-y)\bar{\rho}(t+y)-\bar{\rho}(t-y)\bar{\sigma}(t+y)=
\res_\lambda \lambda^{-1}\bar{\tau}_z(t-y-2[\lambda])\bar{\tau}(t+y+2[\lambda])\e^{\xi(-2y,\lambda)} \nonumber\\
&-\res_\lambda \lambda^{-1}\bar{\tau}(t-y-2[\lambda])(\partial_z \log\bar{\tau}(t-y))\bar{\tau}(t+y+2[\lambda])\e^{\xi(-2y,\lambda)},\\
&2\bar{\sigma}(t-y)\bar{\tau}(t+y)-\bar{\sigma}(t+y)\bar{\tau}(t-y)
=\res_\lambda \lambda^{-1}\bar{\tau}(t-y-2[\lambda])
\bar{\sigma}(t+y+2[\lambda]) \e^{\xi(-2y,\lambda)},\nonumber\\
\\
&2\bar{\rho}(t-y)\bar{\tau}(t+y)-\bar{\rho}(t+y)\bar{\tau}(t-y)
=\res_\lambda \lambda^{-1}\bar{\tau}(t-y-2[\lambda])
\bar{\rho}(t+y+2[\lambda]) \e^{\xi(-2y,\lambda)}.\nonumber\\
\end{align}
\end{subequations}
Here we define $\bar{f}(t):= f(z-\bar{t_k},t),$
therefore $\bar{f}(t-2[\lambda])= f(z-\bar{t_k}+\frac{2}{k\lambda^k},t-2[ 
\lambda]),$
$\bar{f}(t+2[\lambda])= f(z-\bar{t_k}-\frac{2}{k\lambda^k},t+2[\lambda])$.

Let $t\rightarrow t-y,~t^{'}\rightarrow t+y$,
$y=(y_1,y_3,...)$,  and
introducing the Hirota's operator $D_n$ as in \cite{Hirota}:
$$D_n f(t)\circ g(t)=(\partial_{t_n}-\partial_{t_n^{'}})f(t)g(t^{'})|_{t^{'}_n=t_n},~~
\e^{aD_n}f(t)\circ g(t)=f(t_n+a)g(t_n-a),$$
we can write (\ref{barhirota}) 
as
\begin{subequations}
\label{mhirota}
\begin{align}
&\sum\limits_{j=0}p_j(-2y)p_j(2\tilde{D}) \e^{\sum\limits_{n=1}y_nD_n}
\bar{\tau}(t)\circ\bar{\tau}(t)
=\e^{\sum\limits_{i=1}y_iD_i}\bar{\tau}(t)\circ\bar{\tau}(t),\label{mhirota:a}\\
&\e^{\sum\limits_{i=1}y_iD_i}(\bar{\rho}(t)\circ\bar{\sigma}(t)-\bar{\sigma}(t)\circ\bar{\rho}(t))
=\sum\limits_{j=0}p_j(-2y)p_j(2\tilde{D}) \e^{\sum\limits_{n=1}y_nD_n}
\bar{\tau}(t)\circ\bar{\tau}_z(t) \nonumber\\
&-(\partial_z\log\bar{\tau}(t-y))\sum\limits_{j=0}p_j(-2y)p_j(2\tilde{D})
\e^{\sum\limits_{n=1}y_n D_n}\bar{\tau}(t)\circ\bar{\tau}(t),\label{mhirota:b}\\
&\e^{\sum\limits_{i=1}y_iD_i}(2\bar{\tau}(t)\circ\bar{\sigma}(t)-\bar{\sigma}(t)\circ\bar{\tau}(t))
=\sum\limits_{j=0}p_j(-2y)p_j(2\tilde{D}) \e^{\sum\limits_{n=1}y_n D_n}
\bar{\sigma}(t)\circ\bar{\tau}(t),\label{mhirota:c}\\
&\e^{\sum\limits_{i=1}y_iD_i}(2\bar{\tau}(t)\circ\bar{\rho}(t)-\bar{\rho}(t)\circ\bar{\tau}(t))
=\sum\limits_{j=0}p_j(-2y)p_j(2\tilde{D}) \e^{\sum\limits_{n=1}y_nD_n}
\bar{\rho}(t)\circ\bar{\tau}(t).\label{mhirota:d}
\end{align}
\end{subequations}
where $\tilde{D}=(D_1,\frac{D_3}{3},\frac{D_5}{5},\cdots)$,
 $p_j(t)$ are Schur polynomials defined as $p_j(t)=\sum\limits_{||\alpha||=j}\frac{t^\alpha}{\alpha !}$, whose generating function in general case is given by
    $\e^{\sum\limits_{i=1}^{\infty}t_i\lambda^i}
        =\sum\limits_{j=0}^{\infty}p_j(t)\lambda^j,$
where $\alpha=(\alpha_1,\alpha_3,\cdots),||\alpha||=\sum\limits_{i=1}^{\infty}i\alpha_i,\alpha !=\alpha_1 !\cdot\alpha_3 !\cdots,t^\alpha=t_1^{\alpha_1}t_3^{\alpha_3}\cdots$.

Then (\ref{mhirota}) takes the following form by comparing the powers of $y$
\begin{subequations}
\label{mhirota-y}
\begin{align}
&\frac{D^\gamma}{\gamma !}\bar{\tau}(t)\circ\bar{\tau}(t)
=\sum\limits_{\alpha+\beta=\gamma}\frac{(-2)^{|\alpha|}}{\alpha ! \beta !}p_{||\alpha||}(2\tilde{D})D^\beta
\bar{\tau}(t)\circ\bar{\tau}(t), \label{mhirota-y:a}\\
&\frac{D^\gamma}{\gamma !}(\bar{\rho}(t)\circ\bar{\sigma}(t)-\bar{\sigma}(t)\circ\bar{\rho}(t))
=\sum\limits_{\alpha+\beta=\gamma}\frac{(-2)^{|\alpha|}}{\alpha ! \beta !}p_{||\alpha||}(2\tilde{D})D^\beta
\bar{\tau}(t)\circ\bar{\tau}_z(t) \nonumber\\
&-\sum\limits_{\alpha+\beta+\delta=\gamma}\frac{(-2)^{|\alpha|}}{\alpha ! \beta ! \delta !}
(\partial_y^\delta\partial_z\log\bar{\tau}(t-y)|_{y=0})p_{||\alpha||}(2\tilde{D})D^\beta\bar{\tau}(t)\circ\bar{\tau}(t),\label{mhirota-y:b}\\
&\frac{D^\gamma}{\gamma !}(2\bar{\tau}(t)\circ\bar{\sigma}(t)-\bar{\sigma}(t)\circ\bar{\tau}(t))
=\sum\limits_{\alpha+\beta=\gamma}\frac{(-2)^{|\alpha|}}{\alpha ! \beta !}p_{||\alpha||}(2\tilde{D})D^\beta
\bar{\sigma}(t)\circ\bar{\tau}(t), \label{mhirota-y:c}\\
&\frac{D^\gamma}{\gamma !}(2\bar{\tau}(t)\circ\bar{\rho}(t)-\bar{\rho}(t)\circ\bar{\tau}(t))
=\sum\limits_{\alpha+\beta=\gamma}\frac{(-2)^{|\alpha|}}{\alpha ! \beta !}p_{||\alpha||}(2\tilde{D})D^\beta
\bar{\rho}(t)\circ\bar{\tau}(t). \label{mhirota-y:d}
\end{align}
\end{subequations}

\begin{rmk}
For the case $\gamma=(1,0,0,\cdots)$,
the term with $y=(y_1,0,0,\cdots)$ in
(\ref{mhirota-y:b}) can be written in the following form
$$2D_x\tau_z\circ\tau+D_x\sigma\circ\rho=2D_z\tau_x\circ\tau+D_x\sigma\circ\rho
=D_xD_z\tau\circ\tau+D_x\sigma\circ\rho=0.$$
\end{rmk}

\begin{exm} { [The first type of
(2+1)-dimensional Sawada-Kotera equation with a self-consistent source (2d-SKwS-I)
\cite{Hu,Wu}, i.e., the extended BKP hierarchy (\ref{extended}) with $n=3$ and  $k=5$]
}

Notice that from the definition of $\bar{\tau}$, we know that
$D_{\bar{t_5}}\bar{\tau}\circ\bar{\tau}
=(D_{\bar{t_5}}-D_z)\tau\circ\tau|_{z=z-\bar{t_k}}$.
Therefore, the Hirota's bilinear equations for the 2d-SKwS-I ((\ref{extended}) with $n=3$ and  $k=5$) can be obtained as
\begin{subequations}
\label{Hirota-eq:2d-SKwS-I}
\begin{align}
    &D_xD_z\tau\circ\tau+D_x\sigma\circ\rho=0,
    &\text{by (\ref{mhirota-y:b}) in $y_1$,} \label{Hirota-eq:2d-SKwS-I-a}\\
    &\left[D_x^6-5D_x^3D_{t_3}+9D_x(D_{\bar{t_5}}-D_z)-5D_{t_3}^2\right]
    \tau\circ\tau=0,
    &\text{by (\ref{mhirota-y:a}) in $y_1y_5$,} \label{Hirota-eq:2d-SKwS-I-b}\\
    &(D_x^3-D_{t_3})\sigma\circ\tau=0,
    & \text{by (\ref{mhirota-y:c}) in $y_3$,}\\
    &(D_x^3-D_{t_3})\rho\circ\tau=0,
    & \text{by (\ref{mhirota-y:d}) in $y_3$,}.
\end{align}
\end{subequations}
\end{exm}

\begin{exm}{ [The second type of
(2+1)-dimensional Sawada-Kotera equation with a self-consistent source (2d-SKwS-II)
\cite{Wu}, i.e., the extended BKP hierarchy (\ref{extended}) with $n=5$ and $k=3$]
}

The Hirota's blilinear equations for the 2d-SKwS-II ((\ref{extended}) with $n=5$ and  $k=3$) can be obtained as
\begin{subequations}
\label{Hirota-eq:2d-SKwS-II}
\begin{align}
& D_xD_z\tau\circ\tau+D_x\sigma\circ\rho=0,
&\text{by (\ref{mhirota-y:b}) in $y_1$,} \\
&[D_x^6-5D_x^3(D_{\bar{t_3}}-D_z)+9D_xD_{t_5}-5(D_{\bar{t_3}}-D_z)^2]\tau\circ\tau=0,
&\text{by (\ref{mhirota-y:a}) in $y_1y_5$,}\\
&[D_x^3-(D_{\bar{t_3}}-D_z)]\sigma\circ\tau=0,
&\text{by (\ref{mhirota-y:c}) in $y_3$,}\\
&[D_x^3-(D_{\bar{t_3}}-D_z)]\rho\circ\tau=0,
&\text{by (\ref{mhirota-y:d}) in $y_3$,}\\
&[D_x^5+5D_x^2(D_{\bar{t_3}}-D_z)-6D_{t_5}]\sigma\circ\tau=0,
&\text{by (\ref{mhirota-y:c}) in $y_5$,}\\
&[D_x^5+5D_x^2(D_{\bar{t_3}}-D_z)-6D_{t_5}]\rho\circ\tau=0,
&\text{by (\ref{mhirota-y:d}) in $y_5$.}
\end{align}
\end{subequations}
\end{exm}

\section{Back to Nonlinear Equations from Hirota's Bilinear Equations}
\label{Transformation from Hirota's bilinear to nonlinear form}

Following the way given in \cite{Hirota}, we can convert the bilinear equations back to nonlinear PDEs.
Consider the following identities, which are easy to prove
\begin{subequations}
\label{back}
\begin{align}
&\e^{\sum\limits_i \delta_iD_i}\rho\circ\tau=\e^{2\cosh(\sum\limits_i \delta_i\partial_i)\log\tau}\e^{\sum\limits_i\delta_i\partial_i}(\rho/\tau),\\
&\cosh(\sum\limits_i \delta_iD_i)\tau\circ\tau=\e^{2\cosh(\sum\limits_i \delta_i\partial_i)\log\tau}.
\end{align}
\end{subequations}

Using transformations $u:=2\partial_x^2\log\tau$, $r:=\rho/\tau$, $q:=\sigma/\tau$
and expanding (\ref{back}) with respect to $\delta$, we can get the relation between the Hirota bilinear equations
and the usual nonlinear form. Here we give two examples.

\begin{exm}
The Hirota's bilinear equations (\ref{Hirota-eq:2d-SKwS-I})
can be translated back to a nonlinear PDEs as
\begin{align*}
&\partial^{-1}\partial_z u +q_x r-q r_x=0,\\
&u^{(5)}+15u_{x}u^{(2)}+15uu^{(3)}+45u^2u_{x}-5u^{(2)}_{t_3}
    -15u_{x}(\partial^{-1}u_{t_3})-15uu_{t_3}\nonumber\\
&+9u_{\bar{t_5}}-9u_z-5\partial^{-1}u_{t_3t_3}=0,\\
&q_{t_3}=q_{xxx}+3uq_x,\\
&r_{t_3}=r_{xxx}+3ur_x.
\end{align*}
After eliminating variable $z$ from above equations, we get the first type of
(2+1)-dimensional Sawada-Kotera equation with a self-consistent source (2d-SKwS-I) \cite{Hu,Wu}
\begin{subequations}
\label{eq:2d-SKwS-I}
\begin{align}
&u_{\bar{t_5}}+\frac{1}{9}u^{(5)}+\frac{5}{3}u_{x}u^{(2)}
 +\frac{5}{3}uu^{(3)}+5u^2u_{x}-\frac{5}{9}u^{(2)}_{t_3}
 -\frac{5}{3}u_{x}(\partial^{-1}u_{t_3})-\frac{5}{3}uu_{t_3}\nonumber\\
 &-\frac{5}{9}\partial^{-1}u_{t_3t_3}+q_{xx}r-qr_{xx}=0,\\
&q_{t_3}=q_{xxx}+3uq_x,\\
&r_{t_3}=r_{xxx}+3ur_x.
\end{align}
\end{subequations}
The Hirota's bilinear equations for the 2d-SKwS-I (\ref{eq:2d-SKwS-I}) is given by (\ref{Hirota-eq:2d-SKwS-I}), which coincide with the form given in \cite{Hu}
after eliminating variable $z$ by (\ref{Hirota-eq:2d-SKwS-I-a}) and (\ref{Hirota-eq:2d-SKwS-I-b}).
\end{exm}

\begin{exm}
The Hirota's bilinear equations (\ref{Hirota-eq:2d-SKwS-II}) can be translated back to a nonlinear PDEs as
\begin{align*}
&\partial^{-1}\partial_z u +q_x r-q r_x=0,\\
&u^{(5)}+15u_{x}u^{(2)}+15uu^{(3)}+45u^2u_{x}-5u^{(2)}_{\bar{t_3}}
    -15u_{x}(\partial^{-1}u_{\bar{t_3}})-15uu_{\bar{t_3}}\\
&+9u_{t_5}+5u_z^{(2)}+15u_{x}\partial^{-1}u_z+15uu_z-5\partial^{-1}(\partial_{\bar{t_3}}-\partial_z)^2u=0,\\
&q_{\bar{t_3}}-q_z=q_{xxx}+3uq_x,\\
&r_{\bar{t_3}}-r_z=r_{xxx}+3ur_x,\\
&q^{(5)}+10uq^{(3)}+5(u_{xx}+3u^2)q_x+5(q_{\bar{t_3}}-q_z)^{(2)}+5u(q_{\bar{t_3}}-q_z)\\
&+10q_x\partial^{-1}u_{\bar{t_3}}-10q_x\partial^{-1}u_z-6q_{t_5}=0,\\
&r^{(5)}+10ur^{(3)}+5(u_{xx}+3u^2)r_x+5(r_{\bar{t_3}}-r_z)^{(2)}+5u(r_{\bar{t_3}}-r_z)\\
&+10r_x\partial^{-1}u_{\bar{t_3}}-10r_x\partial^{-1}u_z-6r_{t_5}=0.
\end{align*}
After eliminating variable $z$ and $q_{\bar{t_3}},r_{\bar{t_3}}$,
we get the second type of
(2+1)-dimensional Sawada-Kotera equation with a self-consistent source (2d-SKwS-II) \cite{Wu}
\begin{subequations}
\label{eq:2d-SKwS-II}
\begin{align}
&u_{t_5}+\frac{1}{9}u^{(5)}+\frac{5}{3}u_{x}u^{(2)}
 +\frac{5}{3}uu^{(3)}+5u^2u_{x}-\frac{5}{9}u^{(2)}_{\bar{t_3}}
 -\frac{5}{3}u_{x}(\partial^{-1}u_{\bar{t_3}})-\frac{5}{3}uu_{\bar{t_3}}
 -\frac{5}{9}\partial^{-1}u_{\bar{t_3}\bar{t_3}}\nonumber\\
 &=\frac{1}{9}[10q^{(4)}r+5q^{(3)}r_{x}-5q_{x}r^{(3)}
    -10qr^{(4)}+5(q_{x}r-qr_{x})_{\bar{t_3}}\nonumber\\
 &+30u(q^{(2)}r-qr^{(2)})+30u_{x}(q_{x}r-qr_{x})],\\
&q_{t_5}=q^{(5)}+5uq^{(3)}+5u^{'}q^{(2)}+[\frac{10}{3}u^{(2)}+5u^2+\frac{5}{3}\partial^{-1}u_{\bar{t_3}}+\frac{5}{3}(q_xr-qr_x)]q_x,\\
&r_{t_5}=r^{(5)}+5ur^{(3)}+5u^{'}r^{(2)}+[\frac{10}{3}u^{(2)}+5u^2+\frac{5}{3}\partial^{-1}u_{\bar{t_3}}+\frac{5}{3}(q_xr-qr_x)]r_x.
\end{align}
\end{subequations}
The Hirota's bilinear equations for the 2d-SKwS-II (\ref{eq:2d-SKwS-II}) is given by (\ref{Hirota-eq:2d-SKwS-II}). To the best of our knowledge,
the Hirota's bilinear form (\ref{Hirota-eq:2d-SKwS-II}) for (\ref{eq:2d-SKwS-II})
 did not appear in the literatures before.
\end{exm}

\section{Conclusion and discussions}

In this paper, by introducing an auxiliary $\partial_z-$flow,
we constructed the bilinear identities (\ref{res-extented}) for the extended BKP hierarchy (\ref{extended}) introduced in \cite{Wu}. The bilinear identities (\ref{res-extented}) are used to generate all the Hirota's bilinear equations
for the extended BKP hierarchy (\ref{extended}).
As examples, the bilinear forms for
the two types of (2+1)-dimensional Sawada-Kotera equation with a self-consistent source (2d-SKwS-I and 2d-SKwS-II) are derived.
The correctness of these bilinear forms are affirmed
by translating the bilinear equations back to the nonlinear PDEs.
The Hirota's bilinear equations for the 2d-SKwS-II is given explicitly,
which did not appear in the literatures before .

It is a very interesting problem to consider the quasi-periodic solutions for the extended BKP hierarchy as we have already obtained its bilinear identities. Another interesting problem is  to consider the bilinear identities for some other extended hierarchies, such as 2D Toda, discrete KP, etc. We will investigate these problems in future.

\section*{Acknowledgement}
This work is supported by National Natural Science Foundation of China (11471182, 11171175, 11201477).

\end{document}